\documentclass[12pt, journal, onecolumn]{IEEEtran}

\usepackage{longtable}
\usepackage{rotating}
\usepackage{multirow}
\usepackage{epsfig, amsmath,amssymb,epsf,cite,subfigure,scalefnt,multirow,array,setspace}
\usepackage{algorithm}
\usepackage{algorithmic}
\usepackage{times}
\usepackage{epstopdf}
\usepackage{subfig}
\usepackage{amsfonts}
\usepackage{booktabs}
\usepackage{multirow}
\usepackage{multicol}
\usepackage{cite}
\usepackage{algorithm}
\usepackage{algorithmic}
\usepackage{stfloats}
\usepackage{bm}
\usepackage{graphicx}
\usepackage{color}
\usepackage{enumerate}

\usepackage[final]{pdfpages}
\usepackage{fancyhdr}
\usepackage{lineno}

  \def\cC{{\mathcal{C}}}

 \def\cN{{\mathcal{N}}}

\def\ba{{\mathbf{a}}}    
    
   \def\bn{{\mathbf{n}}} 
  \def\br{{\mathbf{r}}} \def\bs{{\mathbf{s}}} 
   \def\bx{{\mathbf{x}}} 
\def\bz{{\mathbf{z}}} \def\bh{\mathbf{h}}

\def\bA{{\mathbf{A}}} \def\bB{{\mathbf{B}}}   
 \def\bG{{\mathbf{G}}} \def\bH{{\mathbf{H}}} \def\bI{{\mathbf{I}}} 
   \def\bN{{\mathbf{N}}} 
  \def\bR{{\mathbf{R}}} \def\bS{{\mathbf{S}}} 
 \def\bV{{\mathbf{V}}}

\def\argmin{\mathop{\mathrm{argmin}}}
\def\argmax{\mathop{\mathrm{argmax}}}

\def\tr{\mathop{\mathrm{tr}}}

\def\span{\mathop{\mathrm{span}}}
\def\vec{\mathop{\mathrm{vec}}}

\def\diag{\mathop{\mathrm{diag}}}

\renewcommand{\baselinestretch}{1.0}

     \def\d4{\!\!\!\!}

     \def\bthe{\boldsymbol{\theta}}

\def\bThe{\boldsymbol{\Theta}}

   \def\C{{\mathbb{C}}}   \def\E{{\mathbb{E}}}

           \def\lA{\left\|}     \def\rA{\right\|}

  \def\-{\! - \!}  \def\+{\! + \!}  \def\={\! = \!}  \def\>{\! > \!}

\newtheorem{lemma}{Lemma}

\newtheorem{remark}{Remark}

\newcommand{\bef}{\begin{figure}}
\newcommand{\eef}{\end{figure}}
\newcommand{\beq}{\begin{eqnarray}}
\newcommand{\eeq}{\end{eqnarray}}

\newenvironment{proof}[1][Proof]{\begin{trivlist}
\item[\hskip \labelsep {\bfseries #1}]}{\end{trivlist}}

\newcommand{\qed}{\nobreak \ifvmode \relax \else
\ifdim\lastskip<1.5em \hskip-\lastskip \hskip1.5em plus0em
minus0.5em \fi \nobreak \vrule height0.5em width0.5em
depth0.25em\fi}

\newcommand{\opRe}{\operatorname{Re}}
\newcommand{\opIm}{\operatorname{Im}}
\begin{document}
\begin{spacing}{1.5}
\pagenumbering{arabic}

\title{Cost-Efficient RIS-Aided Channel Estimation via Rank-One Matrix Factorization}
\author{Wei Zhang,~\IEEEmembership{Member,~IEEE}
and Wee Peng Tay,~\IEEEmembership{Senior Member,~IEEE
}
\thanks{
{The authors are with the School of Electrical and Electronic Engineering, Nanyang Technological University, Singapore (e-mail: weizhang@ntu.edu.sg, wptay@ntu.edu.sg).}
}
}
\maketitle
\begin{abstract}
A reconfigurable intelligent surface (RIS) consists of massive meta elements, which can improve the performance of future wireless communication systems.
Existing RIS-aided channel estimation methods try to estimate the cascaded channel directly, incurring high computational and training overhead especially when the number of elements of RIS is extremely large.
In this paper, we propose a cost-efficient channel estimation method via rank-one matrix factorization (MF).
Specifically, if the RIS is employed near base station (BS), it is found that the RIS-aided channel can be factorized into a product of low-dimensional matrices.
To estimate these factorized matrices, we propose alternating minimization and gradient descent approaches to obtain the near optimal solutions.
Compared to directly estimating the cascaded channel, the proposed MF method reduces training overhead substantially.
Finally, the numerical simulations show the effectiveness of the proposed MF method.

\end{abstract}

\section{Introduction}
Due to the increasing bandwidth requirement of modern wireless communication systems, many new technologies such as millimeter wave, hybrid precoding, and massive multi-input and multi-output (MIMO) have been introduced. Though these technologies can provide high spectral efficiency, the hardware cost of these technologies is also large \cite{Heath16}.
Recently, the reconfigurable intelligent surface (RIS) is proposed as an aid to the wireless communications.
An RIS directly reflects the received signal with many low-cost meta elements \cite{EmilBeat2020}, through which the performance of existing systems is improved without high hardware cost.
Another advantage of RIS is that it establishes a reflected transmission path for the communication if the direct transmission is blocked \cite{EmilBeat2020}.
By carefully deploying the RIS near the boundary of a cell, RIS also helps to improve the performance of cell-edge users \cite{PanMul2020}.
Overall, to achieve these benefits introduced by RIS, channel state information is key.


In RIS-aided communication systems,
the signals are  transmitted by the BS to RIS through the BS-RIS link, and reflected passively by RIS, then received by UE through the RIS-UE link.
The channel estimation task is to obtain the estimate of  the cascaded channel which consists of the BS-RIS link and RIS-UE link.
The recent works focusing on estimation of the RIS-aided channel can be categorized into two types.
The first type focuses on estimating the cascaded channel as a whole \cite{Jensen2020,YouGroup2020,chen2019channel}.
The work in \cite{Jensen2020} proposed a minimum mean square error channel estimation approach.
To handle the high training overhead, the work in \cite{YouGroup2020} divided the high-dimensional channel matrix into low-dimensional sub-groups. In \cite{chen2019channel}, a compressed sensing based method was proposed to estimate the angle-of-arrival (AoA) and angle-of-departure (AoD) associated with the RIS-aided channel.
Although the cascaded channel is sufficient to design transmit and passive beamforming \cite{ChenSum2019,Wu2018Intelligent},
the structure of the cascaded channel is not fully considered in \cite{Jensen2020,YouGroup2020,chen2019channel}.
The second type of existing works focus on estimating the channels of BS-RIS link and RIS-UE link \cite{De2021Channel,Hu2021Two,He2020Cascaded}. A tensor decomposition method was proposed in \cite{De2021Channel}, which leverages the parallel factor tensor modeling of the received signals.
The work \cite{Hu2021Two} proposed a two-timescale framework based on the characteristic that the BS-RIS link is quasi-static, while the RIS-UE link is mobile. The authors of \cite{He2020Cascaded} presented a two-stage method that utilizes the techniques of sparse matrix factorization and matrix completion. However, these methods \cite{De2021Channel,Hu2021Two,He2020Cascaded} require high overhead \cite{De2021Channel} or additional system constraints, such as  dual-link pilot transmission in \cite{Hu2021Two} or sparsity of RIS phase shifts in \cite{He2020Cascaded}.


In this work, we develop a cost-efficient method for RIS-aided channel estimation by factorizing the channel matrix as a product of low-dimensional matrices, which  helps to reduce the training overhead without any additional system constraints.
Specifically, we perform rank-one matrix factorization (MF) to formulate the RIS-aided channel estimation as a phase retrieval problem, which is solved by gradient descent or alternating minimization methods.
We  then extend the proposed  MF method into multi-user scenarios. The simulations illustrate that the proposed MF method achieves more accurate estimation than existing works and has low training overhead.

\emph{Notations:} A bold lower case letter $\mathbf{a}$ is a vector and a bold capital letter $\mathbf{A}$ is a matrix. ${{\mathbf{A}}^{T}}$, ${{\mathbf{A}}^{*}}$, ${{\mathbf{A}}^{H}}$, ${{\mathbf{A}}^{-1}}$, $\mathrm{tr}(\mathbf{A})$,  ${{\left\| \mathbf{A} \right\|}_{F}}$ and ${{\left\| \mathbf{a} \right\|}_{2}}$ are, respectively, the transpose, conjugate, Hermitian, inverse, trace, Frobenius norm of $\bA$, and $2$-norm of $\mathbf{a}$.
${{[\mathbf{A}]}_{:,i}}$ and ${{[\mathbf{A}]}_{i,:}}$ are, respectively, the $i$th column and $i$th row of $\mathbf{A}$. $\mathrm{\mathop{vec}}(\mathbf{A})$ stacks the columns of $\mathbf{A}$ and forms a long column vector.
$\mathrm{\mathop{diag}}(\mathbf{a})$ returns a square diagonal matrix with the vector $\ba$ on the main diagonal.
 $\opRe\{z\}$ and  $\opIm\{z\}$ are the real and imaginary parts of the complex number $z$, respectively.
 $\mathbf{A}\otimes \mathbf{B}$ and  $\mathbf{A} \circ \mathbf{B}$ denote the Kronecker and Hadamard product of $\mathbf{A}$ and $\mathbf{B}$, respectively.

\section{Channel Model} \label{section channel model}
In this section, we introduce the signal and channel model of an RIS-aided communication system.

Suppose the BS has $N$ antennas, the RIS has $M$ elements, and the UE has one antenna. The BS-RIS channel is $\bG\in \C^{M \times N}$, and RIS-UE channel is $\bh_r \in \C^{M\times 1}$, and the BS-UE channel is $\bh_d \in \C^{N\times 1}$.
For the RIS, $\bthe \in \C^{M \times 1}$ denotes the phase shifts, i.e.,  $\bthe=[\beta_1e^{j\theta_1},\ldots,\beta_me^{j\theta_m},\ldots,\beta_Me^{j\theta_M}]^T$, where
$\theta_m \in [0,2\pi)$ and $\beta_m \in [0,1]$ are the phase-shift value and amplitude reflection coefficient of the $m$th element. Here, we let $\beta_m=1, \forall m$ to maximize the signal reflection \cite{Jensen2020,Wu2018Intelligent}.

In this work, we assume there is line-of-sight (LOS) path\footnote{The LOS path between RIS and BS can be achieved by deploying RIS near BS. We extend it to other channel models in Section \ref{extension dl MF}.} between BS and RIS. Then, the channel model $\bG$ is given by
{\setlength\abovedisplayskip{0.5pt}
\setlength\belowdisplayskip{0.5pt}
\beq
\bG = \beta_{BR}\ba_R(\phi)\ba_B^H(\psi), \label{channel base RIS}
\eeq
where $\beta_{BR}$ is the complex path gain,
$\phi \in [0,1)$ and $\psi \in [0,1)$ are the effective  AoA and AoD of RIS and BS, respectively. In this work, it is assumed BS and RIS have uniform linear antenna array (ULA)  \cite{He2020Cascaded,chen2019channel}. Thus,
\beq
\!\!\ba_B(\psi) \!\!\!&=&\!\! \!\frac{1}{\sqrt{N}} [ 1, e^{-j2\pi \psi}, \ldots,e^{-j 2 \pi \psi(N-1)} ]^T \in \C^{N\times 1},\nonumber\\
\!\!\ba_R(\phi)\!\!\! &=& \!\!\!\frac{1}{\sqrt{M}} [ 1, e^{-j2\pi \phi}, \ldots,e^{-j 2 \pi \phi(M-1)} ]^T\in \C^{M\times 1},\nonumber
\eeq
are the array response vectors of BS and RIS. The proposed method in this paper can be  extended to the planar antenna array at RIS by expressing $\ba_R(\cdot)$ above in a 2D form \cite{Heath16}.}

For the channel $\bh_r$ between RIS and UE, we assume that there is no LOS path and it is the Rayleigh fading channel \cite{NadeemCE2020}, i.e.,
$\bh_r  = \beta_{RU}\bz$,
where $\beta_{RU}$ is the complex path gain, and $\bz \in \C^{M \times 1}$ is a complex Gaussian random variable $\cC \cN(0,1)$.

\vspace{-0.2cm}
\section{RIS-Aided Channel Estimation via MF} \label{section down-link}
In this section, we discuss the down-link and single-user RIS-aided channel estimation by using the MF method.
\vspace{-0.2cm}

\subsection{Problem Formulation}
\vspace{-0.1cm}
In the down-link transmission, the received signal at UE is
{\setlength\abovedisplayskip{0.5pt}
\setlength\belowdisplayskip{0.5pt}
\beq
r=(\bh_d^H+\bh_r^H \bThe \bG)\bx+n,
\eeq
where $\bThe = \diag(\bthe) \in \C^{M \times M}$, $\bx \in \C^{N \times 1}$ is the transmitted signal, and $n \in \C$ is noise distributed as $\cC \cN(0, \sigma^2)$.}

We assume that the channel of direct path $\bh_d$ is  accurately estimated, which can be done by using traditional channel estimation methods before enabling the RIS. Here, we suppose there are $K\ge M$ transmit pilots for channel estimation. Then, the received signal at the $k$th training transmission is
\beq
r_k = \bh_r^H \bThe_k \bG \bx_k + n_k
= \bthe_k^T \diag(\bh_r^H) \bG \bx_k + n_k, \label{original rk}
\eeq
where $\bthe_k$ is the $k$th phase shifts of RIS, and $\bx_k$ is the $k$th transmitted signal with $\| \bx_k\|_2=1$, thus $\text{SNR}=1/\sigma^2$.
The effective cascaded channel is expressed as $\bH_e=\diag(\bh_r^H) \bG \in \C^ {M \times N}$.
Plugging the model of $\bG$ in \eqref{channel base RIS}, we express \eqref{original rk} as
\beq
r_k
= \bthe_k^T \bar{\ba}\ba_B^H(\psi) \bx_k + n_k, \label{revised receive signal}
\eeq
where $\bar{\ba}=\diag(\bh_r^H) \beta_{BR}\ba_R(\phi) \in \C^{M \times 1}$, and $\bH_e = \bar{\ba}\ba_B^H(\psi)$. Then, the task  in this work is to estimate $\psi$ and $\bar{\ba}$ from $\{  r_k\}_{k=1}^K$.
From the model expressed in \eqref{revised receive signal}, the problem formulation of maximal likelihood estimation is,
{\setlength\abovedisplayskip{0.5pt}
\setlength\belowdisplayskip{0.5pt}
\beq
\min_{\bar{\ba},\psi} \sum_{k=1}^{K}\left|\bthe_k^T \bar{\ba}\ba_B^H(\psi) \bx_k -r_k \right|^2. \label{ls problem}
\eeq
For convenience, the objective function in \eqref{ls problem} is defined as $J(\bar{\ba}, \psi)$.}
If we ignore the structure of $\ba_B(\psi)$, the problem in \eqref{ls problem} is a conventional phase retrieval problem, which can be solved by low-rank matrix recovery techniques \cite{ZhangSD,jain2013low,ChiNon2019}. However, these approaches do not take the structure of $\ba_B(\psi)$ into account.
In this work, we aim to estimate $\bar{\ba}$ and $\psi$ instead.
As \eqref{ls problem} is a non-convex problem, we consider the use of alternating minimization and gradient descent approaches to find its solution.
The details of the proposed RIS-aided channel estimation are shown in Algorithm \ref{RIS channel estimation}.
\vspace{-0.3cm}

\begin{algorithm} [t]
\caption{Down-link and single-user RIS-aided channel estimation via matrix factorization}
\label{RIS channel estimation}
\begin{algorithmic} [1]
\STATE Input: BS antennas $N$, RIS elements $M$, transmitted signals $\{\bx_k\}_{k=1}^K$, received signals $\{r_k\}_{k=1}^K$, step size $\eta$.
\STATE Initialization: $\bar{\ba}^{(0)}$ and $\psi^{(0)}$ are from \eqref{initial a solu} and \eqref{initial psi}.
\FOR{$p=0,1,2,\ldots$}
 \STATE Update $\bar{\ba}^{(p)}$, $\psi^{(p)}$ to $\bar{\ba}^{(p+1)}$, $\psi^{(p+1)}$ according to gradient method in \eqref{gradient update a re}, \eqref{gradient update a im} and \eqref{gradient update b} or  the alternating minimization method in \eqref{alter update a} and \eqref{alter update b}.
 \STATE If the iteration statures or the maximum iteration is attained, it terminates the iteration.
\ENDFOR
\STATE Calculate the estimation of cascaded channel: $\widehat{\bH}_e \!=\! \bar{\ba}^{(p+1)} \ba_B^H(\psi^{(p+1)})$.
\STATE Output:  $\widehat{\bH}_e$.
\end{algorithmic}
\end{algorithm}

\vspace{-0.1cm}
\subsection{Initialization}
\vspace{-0.1cm}
Solving the non-convex optimization problem \eqref{ls problem} using either alternating minimization or gradient descent approach requires good initialization.
In the following, we discuss how to initialize these variables.

According to \eqref{revised receive signal}, we define $\bS = {\sqrt{N}}/{K} \sum_{k=1}^{K} r_k \bthe_k^* \bx_k^H \in \C^{M \times N}$.
Then, the initialization of $\psi$ is given by
{\setlength\abovedisplayskip{0.5pt}
\setlength\belowdisplayskip{0.5pt}
\beq
\psi^{(0)} = \argmax_{\psi} \lA  \bS \ba_B(\psi) \rA^2_2, \label{initial psi}
\eeq
which means that $\ba_B(\psi^{(0)} )$ is most correlated with $\bS$.} The problem in \eqref{initial psi} is one-dimensional, which can be solved by using coarse grid search and then refining the result gradually.

To initialize $\bar{\ba}$, we solve \eqref{ls problem} with the initialized $\psi^{(0)}$,
\vspace{-0.2cm}
\beq
\bar{\ba}^{(0)} = \argmin_{\bar{\ba}} \sum_{k=1}^{K}\left|\bthe_k^T \bar{\ba}\ba_B^H(\psi^{(0)}) \bx_k -r_k \right|^2, \label{initilization a}
\eeq
which is a common least squares (LS) problem.
Define $\br=[r_1,\ldots, r_K]^T \in \C^{K \times 1}$ and $\bB^{(0)} \in \C^{K \times M}$ with $[\bB^{(0)}]_{k,:} = \ba_B^H(\psi^{(0)}) \bx_k \bthe_k^T$, the problem in \eqref{initilization a} is rewritten as
{\setlength\abovedisplayskip{1pt}
\setlength\belowdisplayskip{1pt}
\beq
\bar{\ba}^{(0)}= \argmin_{\bar{\ba}}\| \bB^{(0)} \bar{\ba}- \br \|_2^2, \label{initial a solu}
\eeq
where the solution is $\bar{\ba}^{(0)} = ((\bB^{(0)})^H \bB^{(0)})^{-1}(\bB^{(0)})^H \br$.} Because we assume that $K\ge M$, $(\bB^{(0)})^H \bB^{(0)}$ is invertible.
\vspace{-0.4cm}
\subsection{Alternating Minimization}

After initializing the $\psi$ and $\bar{\ba}$ with  $\psi^{(0)} $ and $\bar{\ba}^{(0)}$, respectively, we iteratively refine these variables as follows:
\begin{itemize}
  \item Fix $\bar{\ba}$, solve the optimization problem \eqref{ls problem} over $\psi$.
  \item Fix $\psi$, solve the optimization problem \eqref{ls problem} over $\bar{\ba}$.
\end{itemize}

We denote $p$ as the index of iteration. For the optimization of $\psi$ with fixed $\bar{\ba}^{(p)}$, the sub-problem is given by
{\setlength\abovedisplayskip{1pt}
\setlength\belowdisplayskip{1pt}
\beq
\psi^{(p+1)} = \argmin_{\psi} \sum_{k=1}^{K}|\bthe_k^T \bar{\ba}^{(p)}\ba_B^H(\psi) \bx_k -r_k |^2. \nonumber
\eeq
By denoting $\bA^{(p)}\in \C^{N \times K}$ with $[\bA^{(p)}]_{:,k} =\bthe_k^T \bar{\ba}^{(p)} \bx_k $, the} problem above can be simplified as
{\setlength\abovedisplayskip{0.5pt}
\setlength\belowdisplayskip{0.5pt}
\beq
\psi^{(p+1)} = \argmin_{\psi}\|\ba_B^H(\psi) \bA^{(p)} -\br^T  \|_2^2,\label{alter update b}
\eeq
which can be solved by using similar technique as \eqref{initial psi}.}
For optimizing $\bar{\ba}$ with fixed $\psi^{(p+1)}$, the sub-problem is
$\bar{\ba}^{(p+1)}= \argmin_{\bar{\ba}} \sum_{k=1}^{K}| \bthe_k^T \bar{\ba} \ba_B^H(\psi^{(p+1)}) \bx_k -r_k |^2$, whose solution is
\beq
 \bar{\ba}^{(p+1)} = ((\bB^{(p+1)})^H \bB^{(p+1)})^{-1}(\bB^{(p+1)})^H \br, \label{alter update a}
\eeq
with $[\bB^{(p+1)}]_{k,:} = \ba_B^H(\psi^{(p+1)}) \bx_k \bthe_k^T$.
\vspace{-0.2cm}
\subsection{Gradient Descent}
Apart from the alternating minimization, we can also employ a gradient descent approach to obtain the local optimal solution of \eqref{ls problem}. The gradient $J(\bar{\ba}, \psi)$ in \eqref{ls problem} with respect to the real and imaginary parts of $\bar{\ba}$ are given by
{\setlength\abovedisplayskip{1pt}
\setlength\belowdisplayskip{1pt}
\beq
\frac{dJ}{d \opRe \{\bar{\ba}\}} =  \opRe \left\{\sum_{k=1}^{K} 2 \bx_k^H \ba_B(\psi) \bthe_k^*(\ba_B^H(\psi)  \bx_k \bthe_k^T \bar{\ba}-r_k)\right\}, \nonumber \\
\frac{dJ}{d \opIm\{\bar{\ba}\}} =  \opIm \left\{\sum_{k=1}^{K} 2 \bx_k^H \ba_B(\psi) \bthe_k^*(\ba_B^H(\psi)  \bx_k \bthe_k^T \bar{\ba}-r_k) \right\}, \nonumber
\eeq
respectively.} From the chain rule, the gradient  of $J(\bar{\ba}, \psi)$ with respect to $\psi$ is given by
{\setlength\abovedisplayskip{0.5pt}
\setlength\belowdisplayskip{0.5pt}
\beq
\frac{dJ}{d\psi}\! =\!\left( \frac{dJ}{d \opRe \{\ba_B\}}\right)^T\!\! \frac{d\opRe\{ \ba_B\}}{d\psi}\!\!+\!\!\left( \frac{dJ}{d\opIm \{\ba_B\}}\right)^T\!\! \frac{d\opIm \{\ba_B\}}{d\psi}, \nonumber
\eeq}
{\setlength\abovedisplayskip{0.5pt}
\setlength\belowdisplayskip{0.5pt}where
\beq
\frac{dJ}{d \opRe \{\ba_B\}} &\!\! =&\!\! \opRe  \left\{\sum_{k=1}^{K} 2\bthe_k^T \bar{\ba}\bx_k (\bar{\ba}^H \bthe_k^* \bx_k^H \ba_B -r_k^H) \right\}, \nonumber \\
\frac{dJ}{d \opIm \{\ba_B\}} &\!\! =&\!\!  \opIm  \left\{\sum_{k=1}^{K} 2\bthe_k^T \bar{\ba}\bx_k (\bar{\ba}^H \bthe_k^* \bx_k^H \ba_B -r_k^H)\right\}, \nonumber\\
\frac{d\opRe \{\ba_B\}}{d\psi} &\!\! =&\!\! \frac{-1}{\sqrt{N}} \sin(\psi \bz) \circ \bz,
~\frac{d\opIm \{\ba_B\}}{d\psi} \!= \!\frac{-1}{\sqrt{N}} \cos(\psi \bz) \circ \bz,\nonumber
\eeq
with $\bz=2\pi[0,1,\ldots,N-1]^T$.} Therefore, the updating rules of $\bar{\ba}$ and $\psi$ are explicitly,
{\setlength\abovedisplayskip{0.6pt}
\setlength\belowdisplayskip{0.6pt}
\beq
\opRe \{\bar{\ba}^{(p+1)}\} &=&\opRe \{\bar{\ba}^{(p)}\}  \!-\! \eta \frac{dJ}{d \opRe \{\bar{\ba}\}}\bigg|_{\bar{\ba}=\bar{\ba}^{(p)}}, \label{gradient update a re} \\
\opIm \{\bar{\ba}^{(p+1)}\} &=& \opIm \{\bar{\ba}^{(p)}\}\!  -\! \eta \frac{dJ}{d \opIm\{ \bar{\ba}\}}\bigg|_{\bar{\ba}=\bar{\ba}^{(p)}},\label{gradient update a im}  \\
\psi^{(p+1)} &=& \psi^{(p)} -\eta \frac{dJ}{d\psi}\bigg|_{\psi=\psi^{(p)}} \label{gradient update b},
\eeq
where $\eta$ is the step size of gradient descent.}

\vspace{-0.4cm}

\subsection{Discussions of the Proposed MF Method}\label{extension dl MF}
In this subsection, we talk about the extension of proposed MF method when the channel model has the different structures as provided Section \ref{section channel model}.
\begin{remark}\label{remark more rank}
When there are more than one paths in BS-RIS link, i.e., $\bG = \sum_{l=1}^{L}\beta^{(l)}_{BR}\ba_R(\phi^{(l)})\ba_B^H(\psi^{(l)})$, where $L$ is the number of paths, the cascaded channel can be expressed as
{\setlength\abovedisplayskip{1pt}
\setlength\belowdisplayskip{1pt}
\beq
\bH_e = \sum_{l=1}^{L}\diag(\bh_r^H) \beta^{(l)}_{BR}\ba_R(\phi^{(l)})\ba_B^H(\psi^{(l)})= \sum_{l=1}^{L}\bH_e^{(l)}. \nonumber
\eeq
Each $\bH_e^{(l)}$ can be estimated by the proposed MF method after subtracting the contributions of the previously estimated paths.}
\end{remark}
\begin{remark}
 When the RIS is located near the UE resulting in a LOS path between RIS and UE and no LOS path between BS and RIS. The channel of RIS-UE link is $\bh_r =  \beta_{RU} \ba_R(\phi)$, where $\beta_{RU}$ is the complex path gain and $\phi$ is the AoD of RIS. The channel $\bG\in \C^{M \times N}$ of BS-RIS link is Rayleigh fading. Thus, the received signal at $k$th transmission in \eqref{original rk} is
{\setlength\abovedisplayskip{0.5pt}
\setlength\belowdisplayskip{0.5pt}
\beq
r_k  &=& (\beta_{RU} \ba_R(\phi) )^H \bThe_k \bG \bx_k + n_k  \nonumber \\
  &=&  \ba_R(\phi) ^H \bThe_k \bar{\bG} \bx_k + n_k, \nonumber
\eeq
where $ \bar{\bG} = \beta_{RU} \bG$. In this scenario, the channel problem is to estimate} $\phi$ and $\bar{\bG}$. The proposed MF methods can also be applicable. It is worth noting that the number of minimal required training pilots is $MN$ because $\bar{\bG}$ is full rank.
\end{remark}

\begin{algorithm} [t]
\caption{Up-link and multi-user RIS-aided channel estimation via matrix factorization}
\label{RIS channel estimation up-link}
\begin{algorithmic} [1]
\STATE Input: BS antennas $N$, RIS elements $M$, number of users $Q$, transmitted signals $x_{q,k,t}, \forall q\le Q, k\le K, t \le T$, received signals $\{\bR_k\}_{k=1}^K$, phase shift $\{\bthe_k\}_{k=1}^{K}$.
\STATE Obtain the signal of $q$th UE $\bS_q$ through \eqref{q observation}.
\STATE Estimate $\widehat{\psi}$ of the BS-RIS link through \eqref{up-link estimate psi}.
\FOR{$q=1,2,\ldots,Q$}
\STATE Estimate $\{\hat{\bar{\ba}}_q\}_{q=1}^Q$ of the RIS-UE link through \eqref{up-link estimation aq}.
\STATE Calculate the estimation of cascaded channel of $q$th UE: $\widehat{\bH}_q = \ba_B(\widehat{\psi})\hat{\bar{\ba}}_q^H$.
\ENDFOR
\STATE Output:  $\{ \widehat{\bH}_q\}_{q=1}^Q$.
\end{algorithmic}
\end{algorithm}

\vspace{-0.2cm}
\section{Extension to Multi-User Scenarios} \label{section up-link}
In the down-link and multi-user RIS-aided systems, each user can estimate the channel separately, which is a trivial extension of the method discussed in Section \ref{section down-link}.
Therefore, in this section, we mainly focus on the up-link and multi-user RIS-aided channel estimation by using the MF method.
\vspace{-0.4cm}
\subsection{Problem Formulation}
Suppose there are $Q$ users, the channel between BS and RIS is $\bG \in \C^{N \times M}$, the channel between the RIS and UE $q$ is $\bh_q \in \C^{M \times 1}$, which are all defined in a similar way as down-link scenario in Section \ref{section down-link}.
The received signal at the BS is  the summation of $Q$ UEs, i.e.,
$\br= \sum_{q=1}^{Q}\bG\diag(\bthe) \bh_q x_q+
\bn   \in \C^{N \times 1}$,
where $|x_q|=1$ denotes the transmitted signal of $q$th UE and $\bn \in \C^{N\times 1}$ is noise distributed as $\cC \cN(\boldsymbol{0}, \sigma^2\bI_N)$.
For the channel estimation framework, it is assumed that there are $K$ transmitted blocks. In each block, the RIS fixes the phase shifts, and each UE sends $T$ symbols.
Preciesly, in the $k$th block, the received signal at BS is $\bR_k \in \C^{N \times T}$, given by
{\setlength\abovedisplayskip{0.5pt}
\setlength\belowdisplayskip{0.5pt}
\beq
\bR_k= \sum_{q=1}^{Q}\bG\diag(\bthe_k) \bh_q \bx_{q,k}^T +\bN_k, \label{obs Rt}
\eeq
where $\bx_{q,k}  \! =  \![x_{q,k,1},\ldots,x_{q,k,T}]^T \!  \!\in    \!\C^{T \times 1}$ is} the transmitted signal of $q$th UE in $k$th block, $\bthe_k \in \C^{M \times 1}$ denotes the $k$th RIS phase shifts, and $\bN_k=[\bn_{k,1},\ldots,\bn_{k,T}]\!  \in    \! \C^{N \times T}$ denotes the noise in $k$th block.

We assume the signals transmitted by different UEs are orthogonal, i.e.,
$\bx_{q,k}^H \bx_{p,k} =0, \forall q\neq p$ and $\bx_{q,k}^H \bx_{q,k} =T, \forall q$, which can be achieved if $T \ge Q$.
Thus, we right multiply $\bx_{q,k}^*$ for \eqref{obs Rt} to obtain the signal associated with $q$th UE as
{\setlength\abovedisplayskip{0.5pt}
\setlength\belowdisplayskip{0.5pt}
\beq
 \bR_k \bx_{q,k}^* =T \bG\diag(\bthe_k)\bh_q +\bN_k  \bx_{q,k}^*. \label{mid form}
\eeq
By defining $\bs_{q,k}  =(1/T) \bR_k\bx_{q,k}^* \in \C^{N \times 1}$ and $\bn_{q,k} =(1/T) \bN_k  \bx_{q,k}^* \in \C^{N \times 1}$}, we rewrite the expression in \eqref{mid form} as
{\setlength\abovedisplayskip{0.5pt}
\setlength\belowdisplayskip{0.5pt}
\beq
~~~~~\bs_{q,k} =\bG\diag(\bthe_k)\bh_q + \bn_{q,k} =\bG \diag(\bh_q)\bthe_k + \bn_{q,k}, \nonumber
\eeq
where the entries in $\bn_{q,k}$ are i.i.d. with distribution of $\cC \cN(0,\sigma^2/T)$.}
We then collect $K$ blocks $\{ \bs_{q,k}\}_{k=1}^K$ of the $q$th UE as $\bS_q =[\bs_{q,1},\ldots,\bs_{q,K}]\!\in\! \C^{N \times K}$ in the following,
{\setlength\abovedisplayskip{0.5pt}
\setlength\belowdisplayskip{0.5pt}
\beq
\bS_q &=& \bG \diag(\bh_q) [\bthe_1,\ldots,\bthe_K]  +[\bn_{q,1},\ldots,\bn_{q,K}]\nonumber \\
&=& \bG \diag(\bh_q) \bar{\bThe}  + \bN_q, \label{q observation}
\eeq
where $\bar{\bThe} = [\bthe_1,\ldots,\bthe_K] \in \C^{M \times K}$}, and $\bN_q = [\bn_{q,1},\ldots,\bn_{q,K}]\in \C^{M \times K}$.
Therefore, from \eqref{q observation}, the cascaded channel for $q$th UE is $\bH_q = \bG \diag(\bh_q) \in \C^{N \times M}$.
Since $\bG = \beta_{BR} \ba_B(\psi) \ba_R^H (\phi)$ in \eqref{channel base RIS},
we express $\bS_q$ in \eqref{q observation} as
{\setlength\abovedisplayskip{0.5pt}
\setlength\belowdisplayskip{0.5pt}
\beq
\bS_q = \ba_B(\psi) \bar{\ba}_q^H \bar{\bThe} +  \bN_q, \label{exp Sq}
\eeq
where $\bar{\ba}_q= (\beta_{BR}\ba_R^H (\phi) \diag(\bh_q))^H \in \C^{M \times 1} $.}
Then, the cascaded channel for $q$th UE is given by $ \bH_q=\ba_B(\psi) \bar{\ba}_q^H$.
Therefore, the channel estimation task is to estimate $\psi$ and $\{ \bar{\ba}_q\}_{q=1}^Q$. Because of the Guassian distribution of $\bN_q$, we consider the following optimization problem:
{\setlength\abovedisplayskip{1pt}
\setlength\belowdisplayskip{1pt}
\beq
\min_{\psi, \{ \bar{\ba}_q\}_{q=1}^Q} \sum_{q=1}^{Q} \| \bS_q - \ba_B(\psi) \bar{\ba}_q^H \bar{\bThe}\|_F^2. \label{up-link form}
\eeq}

In the following, we complete the estimation of $\psi$ and $\{ \bar{\ba}_q\}_{q=1}^Q$ from \eqref{up-link form} in two steps. In the first step, we estimate
the channel associated with BS-RIS link, i.e., $\psi$, which is shared among multiple users.
In the second step, we estimate channel associated with the  RIS-UE link, i.e., $\{ \bar{\ba}_q\}_{q=1}^Q$.
\vspace{-0.4cm}
\subsection{Multi-User Channel Estimation }
\subsubsection{Estimate $\psi$}
We collect the observations of $Q$ UEs as
$\bS = [\bS_1,\ldots,\bS_Q] \in \C^{N \times KQ}$.
It is from \eqref{exp Sq} that $\bS$ is a rank-one matrix, whose column subspace is spanned by $\ba_B (\psi)$. Thus, the estimation of $\psi$ from \eqref{up-link form} is formulated as follows,
{\setlength\abovedisplayskip{0.5pt}
\setlength\belowdisplayskip{0.5pt}
\beq
\widehat{\psi} = \argmax_{\psi} \|  \ba_B(\psi)^H \bS \|_2^2, \label{up-link estimate psi}
\eeq
which can be solved by using similar technique as \eqref{initial psi}.}

\subsubsection{Estimate $\{ \bar{\ba}_q\}_{q=1}^Q$}
With the estimation of $\ba_B(\widehat{\psi})$, solving the problem in \eqref{up-link form} with respect to $\{ \bar{\ba}_q\}_{q=1}^Q$ can then be separated. In particular, for the $q$th UE, we now need to solve the following problem,
{\setlength\abovedisplayskip{1pt}
\setlength\belowdisplayskip{1pt}
\beq
\hat{\bar{\ba}}_q = \argmin_{\bar{\ba}_q}\|  \bS_q - \ba_B(\widehat{\psi})\bar{\ba}_q^H \bar{\bThe}\|_F^2. \label{up-link form q}
\eeq
By vectorizing the notation as}
$\vec(\ba_B(\widehat{\psi}) \bar{\ba}_q^H \bar{\bThe}) = \bar{\bThe}^T \otimes \ba_B(\widehat{\psi}) \vec(\bar{\ba}_q^H)
=\bar{\bThe}^T \otimes \ba_B(\widehat{\psi})\bar{\ba}_q^*$,
the problem in \eqref{up-link form q} can be rewritten as
$
\hat{\bar{\ba}}_q = \argmin_{\bar{\ba}_q} \|  \vec(\bS_q) -\bar{\bThe}^T \otimes \ba_B(\widehat{\psi})\bar{\ba}_q^*  \|_2^2
$.
Since this is a LS problem, and the solution is given by
{\setlength\abovedisplayskip{0.5pt}
\setlength\belowdisplayskip{0.5pt}
\beq
\hat{\bar{\ba}}_q  = \left( (\bV(\widehat{\psi})^H\bV(\widehat{\psi}))^{-1} \bV(\widehat{\psi})^H \vec(\bS_q) \right)\!^*, \label{up-link estimation aq}
\eeq
where $\bV(\widehat{\psi}) = \bar{\bThe}^T \otimes \ba_B(\widehat{\psi}) \in \C^{N K \times M}$.}
We can check that $\| \bV(\widehat{\psi}) \|_F^2 =MK $.
Given $\widehat{\psi}$ in \eqref{up-link estimate psi} and $\hat{\bar{\ba}}_q$ in \eqref{up-link estimation aq}, the estimate of cascaded channel of $q$th UE is $\widehat{\bH}_q = \ba_B(\widehat{\psi})\hat{\bar{\ba}}_q^H$. The details of the multi-user RIS channel estimation are in Algorithm \ref{RIS channel estimation up-link}.

\subsubsection{Design of Phase Shifts of RIS}
By denoting $\bV (\widehat{\psi})^{\dagger} =(\bV(\widehat{\psi})^H\bV(\widehat{\psi}))^{-1}\bV(\widehat{\psi})^H \in \C^{M \times NK}$, we can calculate the mean-square error (MSE) of $\hat{\bar{\ba}}_q $ as follows,
\beq
\E \| \bV (\widehat{\psi})^{\dagger} \! \vec(\bS_q) \! -\! \bar{\ba}_q^*  \|_2^2\!  = \! \E \| \bV (\widehat{\psi})^{\dagger} (\bV (\psi) \bar{\ba}_q^* \!  +  \!\vec(\bN_q)) \! - \!\bar{\ba}_q^*  \| _2^2,  \nonumber
\eeq
where the equality holds from $\bS_q=\bV (\psi) \bar{\ba}_q^* \! + \!\vec(\bN_q)$.
Here, we assume that the estimation of $\psi$ is accurate, so we have $ \bV(\widehat{\psi})^{\dagger} \bV(\psi) \approx \bI_M$. Then, the MSE can be approximated as
\beq
\text{MSE}(\bar{\ba}_q) \! \approx  \! \E \| \bV (\widehat{\psi})^{\dagger}  \!\vec(\bN_q)  \|_2^2   \!=   \! \frac{\sigma^2 }{T} \! \tr((\bV(\widehat{\psi})^{H}\! \bV (\widehat{\psi}))^{-1} ). \! \label{express MSE aq}
\eeq
Therefore, to minimize the MSE with respect to $\bar{\ba}_q$, we need to solve the following problem,
{\setlength\abovedisplayskip{0.5pt}
\setlength\belowdisplayskip{0.5pt}
\beq
\min_{\bar{\bThe}}  \tr((\bV (\widehat{\psi})^{H}\bV (\widehat{\psi}))^{-1}), \text{subject to } \| \bV (\widehat{\psi}) \|_F^2 = MK. \label{min theta}
\eeq
The following lemma states the condition for $\bar{\bThe}$}, which achieves the minimal MSE of $\bar{\ba}_q$ in \eqref{express MSE aq}.

\begin{lemma}
When the phase shift matrix $\bar{\bThe}$ in \eqref{q observation} satisfies $\bar{\bThe} \bar{\bThe}^H = K \bI_M$, the minimal MSE of $\bar{\ba}_q, q=1,\ldots,Q$ in \eqref{express MSE aq} can be achieved, whose value is
$\text{MSE}(\bar{\ba}_q)={\sigma^2M}/{(KT)}$.
\end{lemma}
\begin{proof}
Suppose the eigenvalues of $\bV (\widehat{\psi})^{H}\bV (\widehat{\psi})$ are given by $\lambda_1,\ldots,\lambda_M$. Because we have $\| \bV (\widehat{\psi}) \|_F^2 = MK$, then
$\lambda_1+ ,\ldots,+\lambda_M = MK $.
Thus, the problem in \eqref{min theta} becomes
{\setlength\abovedisplayskip{0.5pt}
\setlength\belowdisplayskip{0.5pt}
\beq
\min \sum_{i=1}^{M} \frac{1}{\lambda_i}, \text{subject to } \sum_{i=1}^{M} \lambda_i=MK, ~\lambda_i \geq 0. \nonumber
\eeq
The above problem is minimized when $\lambda_i  = K, \forall i$, and the minimum is $M/K$.} The optimality condition means that
 $\bV(\widehat{\psi})^H \bV(\widehat{\psi})  = K \bI_M $, precisely,
\beq
(\bar{\bThe}^T \otimes \ba_B(\widehat{\psi}))^H \bar{\bThe}^T \otimes \ba_B(\widehat{\psi}) = K \bI_M. \label{mid condition}
\eeq
Using the fact $\ba_B(\widehat{\psi})^H \ba_B(\widehat{\psi}) = 1$, we can simplify the optimality condition in \eqref{mid condition} as
$\bar{\bThe} \bar{\bThe}^H = K\bI_M$.
Substituting the design of $\bar{\bThe}$ into \eqref{express MSE aq}, the minimized MSE is given by ${\sigma^2M}/{(KT)}$. This concludes the proof.
\end{proof}

\section{Numerical Results} \label{Ssimulation}
 \begin{table}
 \small
		\centering
\caption{Training Overhead of Channel Estimation Methods}
		\begin{tabular}{|c|c|}\hline
			Channel Estimation Methods&Minimal Training Pilots\\ \hline
			Proposed MF-GD&$M$\\ \hline
			Proposed MF-AM &$M$\\ \hline
			LS \cite{Jensen2020}  &$MN$\\ \hline
			LR \cite{jain2013low}  &$M+N$\\ \hline
			KBF \cite{De2021Channel} &  $MN$ \\ \hline
		\end{tabular}
		\label{table use}
	\end{table}

\subsection{Training Pilot Overhead}\label{sec overhead}
In Table \ref{table use}, we compare the minimal training pilots of our proposed MF with gradient descent (MF-GD) and  MF with alternating minimization (MF-AM) with the existing LS method \cite{Jensen2020}, low rank matrix recovery (LR) method \cite{jain2013low}, and KBF method \cite{De2021Channel} in the down-link and single-user scenario. As shown in Table \ref{table use}, the proposed MF-GD and MF-AM only require minimal $M$ training pilots, which are much less than that of the existing works in \cite{Jensen2020,jain2013low,De2021Channel}.

\begin{figure}
\centering
\includegraphics[width=4.0 in]{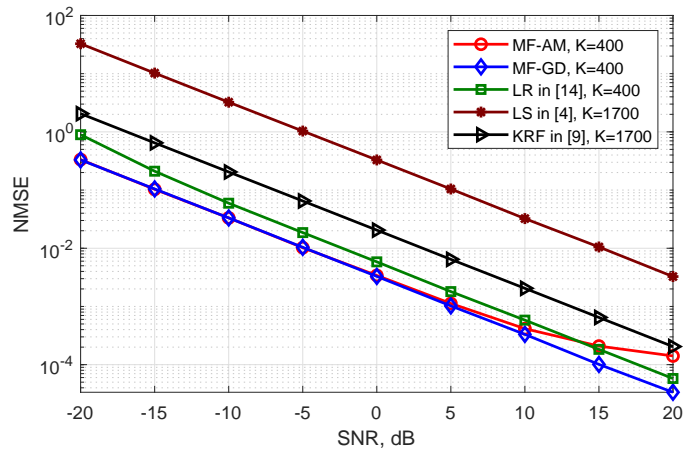}
\caption{NMSE vs. SNR (dB) in down-link and single-user scenario ($N=32, M=50$).} \label{figure_alter_grad}
\end{figure}

\begin{figure}
\centering
\includegraphics[width=4.0 in]{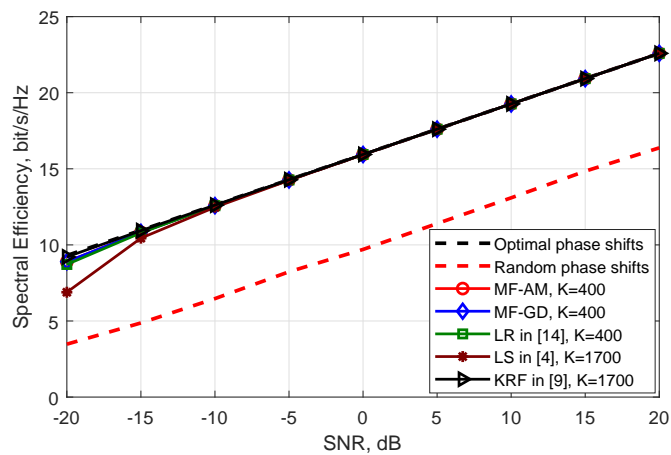}
\caption{Spectral efficiency vs. SNR (dB) in down-link and single-user scenario ($N=32, M=50$).} \label{figure_alter_grad_rate}
\end{figure}
\vspace{-0.3cm}
\subsection{Single-User Scenario}
In Fig. \ref{figure_alter_grad}, we compare the estimation accuracy of our proposed MF-GD and  MF-AM with the LS method \cite{Jensen2020},  LR method \cite{jain2013low}, and KBF method \cite{De2021Channel} in the down-link and single-user scenario.
The normalized MSE (NMSE) is defined as $\E[\| \bH_e  - \widehat{\bH}_e \|_F^2/\| \bH_e  \|_F^2]$.
Since LS  and KBF require $K\ge MN$ to obtain a valid channel estimation, in this work, we let $K=1700$ for LS and KBF methods, but $K  \!=\!400\!\ll\! MN$ for our proposed MF methods.
 Observed from Fig. \ref{figure_alter_grad}, though utilizing much less training pilots, the proposed MF methods still outperform LS and KBF.
 Meanwhile, the proposed MF methods also outperform LR method because they leverage on both the low-rank property and channel structure of the RIS-aided channel.
 Moreover, compared to MF-AM, the MF-GD can achieve even more accurate estimation as SNR increases.

In Fig. \ref{figure_alter_grad_rate}, we compare the spectral efficiency achieved by our proposed MF methods with benchmarks \cite{Jensen2020,jain2013low,De2021Channel}.
The phase shifts of RIS are based on the estimated $\widehat{\bH}_e$ through the method in \cite{Wu2018Intelligent}. The performances of random and optimal phase shifts are also plotted as benchmarks,  where phase shifts of the former are uniformly distributed in $[0,2\pi)$ with zero training pilots, and the latter is designed from the true $\bH_e$. As illustrated in Fig. \ref{figure_alter_grad_rate}, the proposed MF methods can achieve near optimal spectral efficiency with much less training overhead.

\begin{figure}[t]
\centering
\includegraphics[width=4.2 in]{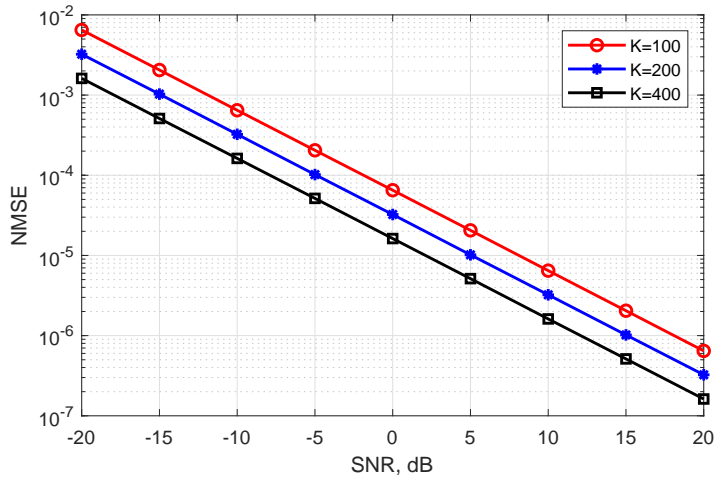}
\caption{NMSE vs. SNR (dB) in up-link and multi-user scenario ($N=32, M=50,Q=5,T=5$).} \label{figure_mul_He}
\end{figure}

\subsection{Multi-User Scenario}
In Fig. \ref{figure_mul_He}, we illustrate the estimation accuracy of proposed MF method in Section \ref{section up-link} for up-link  and multi-user RIS-aided channel.
The NMSE is defined as ${1}/{Q}\sum_{q=1}^Q\E[\| \bH_q  - \widehat{\bH}_q \|_F^2/\| \bH_q  \|_F^2]$.
The simulation parameters are $N=32, M=50, Q=5,T=5$. As can be seen from Fig. \ref{figure_mul_He}, with increasing transmitted blocks $K$, the NMSE of cascaded channel decreases, which is consistent with our analysis.

\vspace{-0.3cm}
\section{Conclusion} \label{SConclusion}
In this paper, we have investigated the RIS-aided channel estimation via MF. By using the proposed MF method, it only estimates low-dimensional matrices to obtain the estimation of RIS-aided channel. Compared to the existing works which directly estimate the cascaded channel, the proposed method achieves more accurate estimation with low training overhead.

\bibliographystyle{IEEEtran}

\bibliography{IEEEabrv,Conference_mmWave_CS}
\clearpage

\end{spacing}
\end{document}